\documentclass[conference]{IEEEtran}
\IEEEoverridecommandlockouts
\usepackage[T1]{fontenc}

\IEEEoverridecommandlockouts                          

\usepackage{amsmath} 
\usepackage{amssymb}  
\usepackage{amsthm}
\usepackage{xcolor}
\usepackage{caption}
\usepackage{graphicx}
\usepackage{mathtools}
\usepackage{caption}
\usepackage{floatrow}
\usepackage{physics} 
\usepackage{tikz}

\usepackage{algorithmicx}
    \usepackage[ruled]{algorithm}
    \usepackage{algpseudocode}

\usepackage{pgfplots}
\usepackage{xcolor}
\usetikzlibrary{matrix,arrows,calc,positioning,shapes,decorations.pathreplacing}
\usepackage{graphicx}

\usepackage{amsmath} 

\usepackage{enumerate}
\usepackage[all,tips]{xy}
\SelectTips{cm}{11}

\newtheorem{theorem}{Theorem}
\newtheorem{corollary}{Corollary}

\newtheorem{remark}{Remark}

\newtheorem{lemma}{Lemma}

\usepackage{bm}

\newcommand{\R}{\mathbb{R}}

\usepackage[noabbrev]{cleveref} 
\crefname{rem}{Remark}{Remarks}
\crefname{exam}{Example}{Examples}
\crefname{assum}{Assumption}{Assumptions}
\crefname{prop}{Proposition}{Propositions}
\crefname{propy}{Property}{Properties}
\crefname{cor}{Corollary}{Corollaries}
\crefname{lem}{Lemma}{Lemmas}
\crefname{section}{Section}{Sections}
\crefname{thm}{Theorem}{Theorems}
\crefname{defn}{Definition}{Definitions}
\crefname{figure}{Fig.}{Fig.}
\Crefname{figure}{Figure}{Figures}
\crefname{equation}{}{}

\usepackage{mathtools}
\usepackage{siunitx}

\begin{document}

\title{Optimal Control of Incompressible Ideal Flows with Obstacle Avoidance\\
\thanks{The authors acknowledge financial support from Grant PID2022-137909NB-C21 funded by MCIN/AEI/ 10.13039/501100011033. A.A.S. was partially supported by MICIU/AEI/10.13039/501100011033/ FEDER, UE, Grant No. PID2024-155187OB-I00.
A.B. was partially supported by NSF grant  DMS-2103026, and AFOSR grants FA
9550-22-1-0215 and FA 9550-23-1-0400. L.C was supported in part by iRoboCity2030-CM, TEC-2024/TEC-62, funded by Comunidad de Madrid.}}

\author{\IEEEauthorblockN{Alexandre Anahory Sim\~oes}
\IEEEauthorblockA{\textit{IE School of Science and Technology,} \\
\textit{IE University}\\
Madrid, Spain \\
alexandre.anahory@ie.edu}
\and
\IEEEauthorblockN{Anthony Bloch}
\IEEEauthorblockA{\textit{Department of Mathematics,} \\ \textit{University of Michigan,}\\
Ann Arbor, MI 48109, USA \\
abloch@umich.edu}
\and
\IEEEauthorblockN{Leonardo Colombo}
\IEEEauthorblockA{\textit{Centre for Automation and Robotics (CSIC-UPM)} \\
Madrid, Spain\\
leonardo.colombo@csic.es}
}

\maketitle


\begin{abstract}
It was shown in \cite{bloch2000optimal} that an optimal control formulation for incompressible ideal fluid flow yields Euler's equations. In this paper, we consider a variational obstacle-avoidance formulation for incompressible ideal flows by introducing a barrier-type potential in the associated optimal control functional. This leads to \textit{modified Euler equations for an inviscid fluid}, in which the barrier term acts on the Lagrangian configuration and appears in the Eulerian description as a shift in the effective pressure. We also present a numerical illustration of the reduced Eulerian dynamics, showing that the barrier term induces a localized deformation of the flow near the obstacle region, consistent with its role as an obstacle-avoidance penalization.
\end{abstract}


\IEEEpeerreviewmaketitle

\section{Introduction}
It is well known that the Euler equations for a perfect fluid can be seen as geodesic equations on the group of volume-preserving diffeomorphisms of the fluid domain, endowed with the $L^2$ Riemannian metric \cite{arnold1966geometrie}. This observation has motivated many developments in geometric mechanics, especially in symmetry and reduction for incompressible fluid flow \cite{marsden1983coadjoint}.

Since the work of Arnold \cite{arnold1966geometrie}, the geometric formulation in terms of diffeomorphism groups has been extended to a broad class of equations arising in hydrodynamics, including compressible fluids \cite{ebin1969groups} and magnetohydrodynamics \cite{marsden1984semidirect}.

Optimal control for incompressible fluids have been studied in \cite{bloch2000optimal} and \cite{holm2009euler}, where the main question is whether Euler's fluid equations admit an optimal control interpretation, and how the resulting fluid equations compare with the symmetric generalized rigid body equations \cite{bloch1996optimal,bloch1997double,bloch1998discrete}. In particular, \cite{bloch2000optimal} shows that an optimal control formulation leads to the standard impulse form of Euler’s equations in Lagrangian variables.

In recent years, path planning has become increasingly important in robotics and, more broadly, in control engineering. The goal is to construct trajectories that satisfy additional requirements such as obstacle avoidance or collision avoidance. An important class of examples arises in rigid body systems, where trajectories may be modeled as curves on $SO(3)\times\R^3$ or $\mathrm{SE}(3)$, both naturally equipped with physically meaningful Riemannian metrics. In this setting, trajectory generation may be viewed as a variational problem on a Lie group.

Geodesic equations, however, do not account for tasks such as obstacle avoidance or inter-agent collision avoidance, which are essential in many path-planning problems. One strategy, introduced in \cite{BlCaCoCDC}, is to augment the action functional with an artificial potential that grows near the obstacles. Necessary conditions for optimality were derived in \cite{BlCaCoCDC,BlCaCoIJC} for obstacle avoidance, and in \cite{mishal,sh} for collision avoidance. Reduction by symmetry on Lie groups and symmetric spaces was considered in \cite{point}.

In this paper, we study a variational optimal control formulation for incompressible ideal fluid flow based on the Hamiltonian--Pontryagin principle. More precisely, we introduce a barrier-type potential in the cost functional in order to penalize fluid configurations approaching a prescribed obstacle region. This leads to \textit{modified Euler equations} on the Lie group of volume-preserving diffeomorphisms. In Theorem~\ref{th1} we derive the corresponding necessary conditions for optimality. Since these equations do not yet provide a closed evolution equation for the Eulerian velocity field, Theorem~\ref{th2} shows that each critical curve induces a modified Euler equation in which the barrier term appears as a shift in the effective pressure. Corollary~1 then gives a boundary interpretation of this pressure shift through the normal pressure balance. We emphasize that the present work focuses on the variational derivation and on the reduced Eulerian interpretation of the barrier term, rather than on the numerical solution of the full two-point boundary value problem.

The paper is organized as follows. Section \ref{sec2} introduces incompressible ideal fluid flow and formulates the obstacle-avoidance problem. Section \ref{sec3} contains the main results on the extremals and the modified Euler equations. Section \ref{sec:exp} presents a numerical illustration of the reduced Eulerian dynamics induced by the barrier-type potential. We conclude with some final remarks and directions for future work. Technical calculations are collected in the Appendix.

\section{Inviscid, Incompressible Fluid Flows} \label{sec2}
\subsection{Dynamical equations for inviscid incompressible flows}



We consider incompressible ideal flows on a domain $\Omega \subset \mathbb{R}^d$ with $d\in\{2,3\}$. In this setting, the Eulerian velocity field $v=v(x,t)$ and the pressure $p=p(x,t)$ satisfy the incompressible Euler equations \cite{arnold1966geometrie}
\begin{equation}\label{eulereq}
    \frac{\partial v}{\partial t} + v \cdot \grad v = - \grad p,\quad
    \hbox{div } v = 0, \,\,\,x \in \Omega. 
\end{equation}
Here $v$ denotes the fluid velocity and $p$ the pressure. The incompressibility condition $\hbox{div } v=0$ expresses preservation of volume along the flow.

Throughout the paper, we assume that $\Omega$ is a bounded domain and that the velocity field is tangent to the boundary. This is the standard geometric setting for incompressible ideal flow on bounded domains; see, for example, \cite{holm1998euler,gaybalmaz2011clebsch}. 



Unlike the whole-space or periodic-box setting considered in some earlier works, here we work on a bounded domain to formulate the obstacle-avoidance penalty in a setting compatible with the geometric description of ideal incompressible flow. The above Eulerian equations will serve as the reference dynamics for the modified equations derived below.

\subsection{Impulse dynamics}


For later use, we recall the impulse formulation associated with the incompressible Euler equations. Let $v$ be a divergence-free velocity field on $\Omega$. We introduce the \emph{impulse density} $z$ by
\begin{equation}
    z = v + \grad k,
    \label{impulsedynamics}
\end{equation}
where $k=k(x,t)$ is a scalar field.


Using \eqref{eulereq} together with \eqref{impulsedynamics}, one obtains
\begin{equation}
    \frac{\partial z}{\partial t} - v \times \curl z = \grad \Lambda,
    \qquad \div v = 0,
    \label{IDmodified}
\end{equation}
where $\displaystyle{\Lambda = \frac{\partial k}{\partial t} - p - \frac{1}{2} v \cdot v}$ is the associated gauge.


At this stage, the scalar field $k$ is not unique, so different gauge choices are possible. Here we adopt the geometric gauge $\Lambda = - v \cdot z$.
With this choice, \eqref{IDmodified} becomes
\begin{equation}
    \frac{\partial z}{\partial t} + (v \cdot \grad) z + (\grad v)^T z = 0,
    \qquad \div v = 0,
    \label{impulseeq}
\end{equation}
and \(k\) is determined by
$\displaystyle{\frac{d k}{d t} = p - \frac{1}{2} v \cdot v}$.

\subsection{Coordinate systems for fluid flows}

It is known, following Arnold \cite{arnold1966geometrie}, that the dynamics of an incompressible ideal fluid admits a geometric interpretation analogous to that of a rigid body. In this setting, the configuration space is the Lie group of volume-preserving diffeomorphisms of the fluid domain $\Omega$, denoted by $\hbox{Diff}_{\mathrm{vol}}(\Omega)$, with group operation given by composition \cite{holm1998euler,gaybalmaz2011clebsch}. 


A configuration $\varphi \in \hbox{Diff}_{\mathrm{vol}}(\Omega)$ maps a reference point $X \in \Omega$ to its current position $x=\varphi(X)\in\Omega$. Thus, $\varphi$ describes the motion of the fluid particles and therefore the fluid configuration.  A motion of the fluid is a curve $\varphi_t \in \hbox{Diff}_{\mathrm{vol}}(\Omega)$, written as $x=\varphi(X,t)$. The associated material velocity is
$\displaystyle{\dot{\varphi}(X,t)=\frac{\partial \varphi(X,t)}{\partial t}}$, and the corresponding Eulerian velocity field is $v(x,t)=\dot{\varphi}(\varphi^{-1}(x,t),t)$. Equivalently,
$v=\dot{\varphi}\circ\varphi^{-1}$.

In what follows, we restrict to the case $\Omega\subset\mathbb{R}^3$, with $\varphi_t$ a volume-preserving diffeomorphism of $\Omega$.





Given a time-dependent scalar field \(k:\Omega\times[0,T]\to\mathbb{R}\) such that
\(k(\cdot,t)\in L^{1}(\Omega)\) for each \(t\in[0,T]\), we write $\displaystyle{\langle k\rangle(t):=\int_{\Omega} k(x,t)\,dx}$. Likewise, given two time-dependent vector fields
\(a,b:\Omega\times[0,T]\to\mathbb{R}^{d}\) such that
\(a(\cdot,t),\,b(\cdot,t)\in L^{2}(\Omega;\mathbb{R}^{d})\) for each \(t\in[0,T]\), we define $\displaystyle{\langle a,b\rangle(t):=\int_{\Omega} a(x,t)^{T}b(x,t)\,dx}$.
Since \(\varphi_t\) is volume-preserving, a change of variables yields
\begin{equation}
    \langle a\circ\varphi, b\rangle
    =
    \langle a, b\circ\varphi^{-1}\rangle.
    \label{change:variables}
\end{equation}


\subsection{Problem formulation}\label{Problem:section}


We now introduce the variational optimal control problem studied in this paper. Given fixed endpoints $\varphi_0,\varphi_T \in \hbox{Diff}_{\mathrm{vol}}(\Omega)$, we consider
\begin{equation}
    \min_{v(\cdot)} \int_0^T \left( \frac{1}{2}\langle v,v\rangle + \int_\Omega V(\varphi(X,t))\,dX \right) dt,
    \label{first:functional}
\end{equation}
subject to
\begin{equation}
    \div v = 0,
    \qquad
    \frac{\partial \varphi}{\partial t} = v \circ \varphi,
    \label{dynamical:constraints}
\end{equation}
and the endpoint conditions $\varphi(X,0)=\varphi_0(X)$ and $\varphi(X,T)=\varphi_T(X)$, where $V:\Omega\to\mathbb{R}$ is a $C^1$ function in the Sobolev space $H^{1}(\Omega;\mathbb{R}^{d})$. 


Here $v$ is the divergence-free Eulerian velocity field, while $\varphi$ is the induced Lagrangian configuration. The term involving $V$ acts on the configuration variable $\varphi$ and is interpreted as a barrier-type potential penalizing configurations that approach a prescribed region of $\Omega$. In this sense, the obstacle avoidance requirement is introduced through the variational formulation. The goal is to derive the corresponding extremal equations and their Eulerian reduced form. Accordingly, the obstacle-avoidance requirement is enforced in a variational sense through penalization, rather than as a state constraint.

\section{Extremals for the obstacle avoidance optimal control problem}\label{sec3}


To derive the extremal equations on $\hbox{Diff}_{\mathrm{vol}}(\Omega)$, we augment the functional \eqref{first:functional} with the constraints \eqref{dynamical:constraints} directly through Lagrange multipliers, in a fashion  similar to that in  \cite{bloch2000optimal,holm2009euler}; see also \cite{gay2013geometric}. Consider the infinite-dimensional space of curves
\begin{align*}
\mathcal{C}
:=
\Big\{(\varphi,v,\pi,k)\ &\Big|\ 
\varphi:\Omega\times[0,T]\to\Omega,\,\varphi(\cdot,t)\in \hbox{Diff}_{\mathrm{vol}}(\Omega),\\
&
\varphi(\cdot,0)=\varphi_0,\,
\varphi(\cdot,T)=\varphi_T,\\
&
v:\Omega\times[0,T]\to\mathbb{R}^3,\, v(\cdot,t)\cdot n=0 \text{ on } \partial\Omega,\\
&
\pi:\Omega\times[0,T]\to\mathbb{R}^3,\,
k:\Omega\times[0,T]\to\mathbb{R}
\Big\}.
\end{align*}

Define the action functional \(\mathcal{J}:\mathcal{C}\to\mathbb{R}\) by
\begin{align*}
\mathcal{J}(\varphi,v,\pi,k)
=
\int_0^T \Big(
&\Big\langle \pi,\,
v\circ\varphi-\frac{\partial\varphi}{\partial t}
\Big\rangle
-\frac12\langle v,v\rangle \\
&\qquad
+\langle k\,\div v + V\circ\varphi\rangle
\Big)\,dt.
\end{align*}

Optimal solutions of the problem in Section~\ref{Problem:section} correspond to critical curves of the functional \(\mathcal{J}\).

\begin{theorem}
\label{th1}
If \((\varphi,v,\pi,k)\) is a critical curve of the functional \(\mathcal{J}\), then it satisfies the following necessary conditions for optimality:
\begin{equation}
\label{necessary:conditions}
\begin{split}
\frac{\partial\varphi}{\partial t} &= v\circ\varphi,\quad \hbox{div } v = 0,\quad
v = \pi\circ\varphi^{-1} - \grad k,\\
\frac{\partial\pi}{\partial t} &= -(Tv\circ\varphi)^T\pi + \grad V\circ\varphi.
\end{split}
\end{equation}
\end{theorem}

\begin{proof}
Let $(\varphi_s,v_s,\pi_s,k_s)$ be a smooth admissible variation of $(\varphi,v,\pi,k)$ in $\mathcal{C}$, with
$
(\varphi_0,v_0,\pi_0,k_0)=(\varphi,v,\pi,k),
$
and denote
$\displaystyle{
(\delta\varphi,\delta v,\delta\pi,\delta k)
=
\left.\frac{d}{ds}\right|_{s=0}
(\varphi_s,v_s,\pi_s,k_s)}$. Since the endpoints are fixed, we have
$
\delta\varphi(\cdot,0)=\delta\varphi(\cdot,T)=0,
$
and since the variation is taken within the class of velocity fields tangent to the boundary,
$
\delta v\cdot n=0$ on $\partial\Omega$.

Differentiating $\mathcal{J}$ with respect to $s$ at $s=0$, we obtain
\begin{equation*}
\begin{split}
0
&=
\delta\mathcal{J}(\varphi,v,\pi,k)(\delta\varphi,\delta v,\delta\pi,\delta k)
\\
&=
\int_0^T
\Big\langle \delta\pi,\,
v\circ\varphi-\frac{\partial\varphi}{\partial t}
\Big\rangle
-
\langle v,\delta v\rangle
+
\langle \delta k\,\div v,\rangle
\\
&\quad
+
\Big\langle \pi,\,
\delta v\circ\varphi+(Tv\circ\varphi)(\delta\varphi)-\frac{\partial\delta\varphi}{\partial t}
\Big\rangle
+
\langle k\,\div\delta v\rangle
\\
&\quad
+
\langle \grad V\circ\varphi,\delta\varphi\rangle
\,dt.
\end{split}
\end{equation*}

Since $\delta\pi$ and $\delta k$ are arbitrary, we immediately obtain
$$
\frac{\partial\varphi}{\partial t}=v\circ\varphi,
\qquad
\div v=0.
$$

Next, using identity \eqref{change:variables}, the term involving $\delta v\circ\varphi$ can be rewritten as
$
\langle \pi,\delta v\circ\varphi\rangle
=
\langle \pi\circ\varphi^{-1},\delta v\rangle.
$
Also, integrating by parts in time and using $\delta\varphi(\cdot,0)=\delta\varphi(\cdot,T)=0$ gives
$$
-\int_0^T \Big\langle \pi,\frac{\partial\delta\varphi}{\partial t}\Big\rangle\,dt
=
\int_0^T \Big\langle \frac{\partial\pi}{\partial t},\delta\varphi\Big\rangle\,dt.
$$

Finally, by the divergence theorem \cite{marsden1993basic},
\begin{align*}
\int_\Omega k(x,t)\,\div\delta v(x,t)\,dx
=&
-\int_\Omega \grad k(x,t)\cdot\delta v(x,t)\,dx\\
&+
\int_{\partial\Omega} k\,\delta v\cdot n\,dS.
\end{align*}
The boundary term vanishes because $\delta v\cdot n=0$ on $\partial\Omega$. Therefore,
\begin{equation*}
\begin{split}
0
=
\int_0^T
&
\langle \pi\circ\varphi^{-1} - v - \grad k,\delta v\rangle
\\
&\quad
+
\Big\langle (Tv\circ\varphi)^T\pi + \frac{\partial\pi}{\partial t} + \grad V\circ\varphi,\delta\varphi\Big\rangle
\,dt.
\end{split}
\end{equation*}

Since $\delta v$ and $\delta\varphi$ are arbitrary admissible variations,
$
v=\pi\circ\varphi^{-1}-\grad k
$
and
$
\frac{\partial\pi}{\partial t}
=
-(Tv\circ\varphi)^T\pi+\grad V\circ\varphi$.
\end{proof}

The equations \eqref{necessary:conditions} are precisely the Euler--Lagrange equations associated with the functional $\mathcal{J}$. However, they do not yet give an explicit evolution equation for the Eulerian velocity field $v$, in contrast with the Euler equations \eqref{eulereq}.

\begin{theorem}\label{th2}
If $(\varphi,v,\pi,k)$ is a critical curve of the functional $\mathcal{J}$, then the Eulerian velocity field $v$ satisfies the modified Euler equations
\begin{equation}\label{mod:eulereq}
\begin{split}
\frac{\partial v}{\partial t}+v\cdot\grad v &= -\grad(p-V),\,\,\,
\div v = 0.
\end{split}
\end{equation}
\end{theorem}

\begin{proof}
The incompressibility condition $\div v=0$ follows from \eqref{necessary:conditions}. It remains to prove the first equation.

Define
$
z=\pi\circ\varphi^{-1}=v+\grad k.
$
Differentiating $z$ with respect to time, we obtain
$$
\frac{\partial z}{\partial t}
=
\frac{\partial \pi}{\partial t}\circ\varphi^{-1}
+
(T\pi\circ\varphi^{-1})\left(\frac{\partial\varphi^{-1}}{\partial t}\right).
$$
Using \eqref{necessary:conditions}, we deduce
\begin{equation*}
\begin{split}
\frac{\partial z}{\partial t}
&=
-(Tv)^T(\pi\circ\varphi^{-1})
+\grad V
+
(T\pi\circ\varphi^{-1})\left(\frac{\partial\varphi^{-1}}{\partial t}\right)\\
&=
-(Tv)^T z+\grad V-Tz(v),
\end{split}
\end{equation*}
where in the last step we used
$
(T\pi\circ\varphi^{-1})\left(\frac{\partial\varphi^{-1}}{\partial t}\right)
=
-Tz(v),
$
see Lemma \ref{aux:lemma:1} in the Appendix. Therefore,
$$
\frac{\partial z}{\partial t}+(Tv)^T z-\grad V+Tz(v)=0.
$$

Using items 1 and 2 of Lemma \ref{aux:lemma:2}, this becomes
$$
\frac{\partial z}{\partial t}
+
(z\cdot\grad)v
+
z\times(\curl v)
-\grad V
+
(v\cdot\grad)z
=0.
$$
Then, by item 3 of Lemma \ref{aux:lemma:2}, we obtain
$$
\frac{\partial z}{\partial t}
+
\grad(v\cdot z)
-
v\times(\curl z)
-
\grad V
=0.
$$

Now let
$
\Lambda=\frac{\partial k}{\partial t}-p-\frac12\,v\cdot v
$
be the gauge associated with $z$. If we choose the geometric gauge
$
\Lambda=-v\cdot z,
$
then the previous equation can be rewritten as
$$
\frac{\partial z}{\partial t}
-
v\times(\curl z)
-
\grad V
=
\grad\left(\frac{\partial k}{\partial t}-p-\frac12\,v\cdot v\right).
$$

On the other hand, as a particular case of item 3 in Lemma \ref{aux:lemma:2}, we have
$
(v\cdot\grad)v
=
\frac12\,\grad(v\cdot v)-v\times(\curl v).
$
Since $v=z-\grad k$ and $\curl(\grad k)=0$, it follows that
$
v\times(\curl z)
=
\frac12\,\grad(v\cdot v)-(v\cdot\grad)v.
$
Hence,
$$
\frac{\partial z}{\partial t}
+
(v\cdot\grad)v
-
\grad V
=
\grad\left(\frac{\partial k}{\partial t}-p\right).
$$

Finally, using
$
\frac{\partial z}{\partial t}
=
\frac{\partial v}{\partial t}
+
\grad\frac{\partial k}{\partial t},
$
we conclude that
$
\frac{\partial v}{\partial t}
+
(v\cdot\grad)v
=
-\grad(p-V)$.
\end{proof}

The previous theorem shows that, after elimination of the auxiliary variables, the extremals satisfy the modified Euler equations \eqref{mod:eulereq}. In particular, the result shows that a configuration-level barrier introduced in the Lagrangian description survives reduction as a pressure-like contribution in the Eulerian equations.


\begin{corollary}
Let $(\varphi,v,\pi,k)$ be a critical curve of $\mathcal{J}$, and define the effective pressure
$\tilde p:=p-V$. Then the modified Euler equations in Theorem~\ref{th2} may be written as
$$
\frac{\partial v}{\partial t}+(v\cdot\grad)v=-\grad\tilde p,
\qquad
\div v=0.
$$
In particular, taking the normal component on $\partial\Omega$ yields
$$\displaystyle{
\partial_n\tilde p
=
-\,n\cdot\left(\frac{\partial v}{\partial t}+(v\cdot\grad)v\right)},
$$
or equivalently,
$$\displaystyle{
\partial_n p
=
\partial_n V
-
n\cdot\left(\frac{\partial v}{\partial t}+(v\cdot\grad)v\right)}.$$
\end{corollary}

\begin{proof}
The first statement follows immediately from Theorem~\ref{th2} by introducing the effective pressure $\tilde p=p-V$. Taking the normal component of
$
\frac{\partial v}{\partial t}+(v\cdot\grad)v=-\grad\tilde p
$ on $\partial\Omega$, we obtain
$\displaystyle{
n\cdot\left(\frac{\partial v}{\partial t}+(v\cdot\grad)v\right)
=
-\,n\cdot\grad\tilde p
=
-\partial_n\tilde p}$. 
Since $\tilde p=p-V$, the final identity follows.
\end{proof}

\begin{remark}
The previous corollary shows that the barrier potential enters the reduced dynamics through an effective pressure shift. Consequently, although the present formulation is open-loop and variational in nature, Theorem~\ref{th2} suggests a possible interpretation in terms of boundary pressure actuation. A full boundary-control formulation, however, lies beyond the scope of the present paper.
\end{remark}

\section{Numerical Illustration}\label{sec:exp}

In this section, we present a numerical illustration of the $2$-dimensional version of equation \eqref{mod:eulereq}. The purpose of this section is not to solve the full boundary value problem associated with the variational formulation, but rather to illustrate the reduced Eulerian dynamics induced by the barrier-type potential. In particular, we show how the modified Euler equations produce a local deformation of the flow in the neighborhood of a circular obstacle region.



To recover a discrete divergence-free vector field after the time discretization, we use a discretized version of the Helmholtz decomposition \cite{marsden_fluids,HDD,PAVLOV}.

Given a vector field on a two-dimensional \(\Delta x\)-uniform grid, denoted by \(X(i,j)=(u(i,j),v(i,j))\), we define the discrete divergence operator by
\begin{equation*}
\begin{split}
\nabla^{d}\cdot X(i,j)
&=
\frac{u(i+1,j)-u(i-1,j)}{2\Delta x}\\
&\quad+
\frac{v(i,j+1)-v(i,j-1)}{2\Delta x}.
\end{split}
\end{equation*}

Given a scalar field on the same grid, denoted by \(\phi(i,j)\), its discrete gradient is defined by
\begin{equation*}
\nabla^{d}\phi(i,j)
=
\left(
\begin{array}{c}
\dfrac{\phi(i+1,j)-\phi(i-1,j)}{2\Delta x}\\[0.8em]
\dfrac{\phi(i,j+1)-\phi(i,j-1)}{2\Delta x}
\end{array}
\right)^{T}.
\end{equation*}

Finally, the discrete Laplacian of a scalar field is given by
\begin{equation*}
\begin{split}
\Delta^{d}\phi(i,j)
&=
\nabla^{d}\cdot \nabla^{d}\phi\\
&=
\frac{\phi(i+2,j)+\phi(i-2,j)}{4\Delta x^{2}}\\
&\quad+
\frac{\phi(i,j+2)+\phi(i,j-2)-4\phi(i,j)}{4\Delta x^{2}}.
\end{split}
\end{equation*}

A discrete divergence-free vector field \(X_{0}\) is a vector field on the grid satisfying $\nabla^{d}\cdot X_{0}=0$.

Given a vector field \(X\) on the grid, we define its discrete Helmholtz projection by $X_{0}=X-\nabla^{d}\phi$, where \(\phi\) is a scalar field satisfying the discrete Poisson equation $\Delta^{d}\phi=\nabla^{d}\cdot X$. With suitable boundary conditions, this scalar field is unique.

We consider the artificial potential
\[
V(x,y)=-A\exp\!\left(-\frac{(x-7)^2+(y-7)^2}{\sigma^2}\right),
\quad A,\sigma>0,
\]
centered at the obstacle location \((7,7)\). Since the reduced modified Euler equation \eqref{mod:eulereq} contains the term \(+\nabla V\), this choice yields an outward shaping term in the neighborhood of the obstacle. In other words, the repulsive character is encoded through the spatial profile of \(V\): the potential is most negative near the obstacle and increases away from it, so that its gradient points radially outward.

The simulations were implemented in Python on a two-dimensional periodic uniform grid. At each time step, we first update the velocity field by a finite-difference discretization of the modified Euler equation, including the barrier contribution through the term $\nabla V$. Since this explicit update does not preserve the divergence-free constraint exactly at the discrete level, we then apply a discrete Helmholtz projection: we solve the discrete Poisson equation associated with the discrete divergence of the updated field and subtract the corresponding discrete gradient. This yields an approximately divergence-free velocity field at each step, which is then used in the next iteration. 
We use the initial velocity field $u(x,y,0)=-\sin y \cos x$, $v(x,y,0)=\sin y \cos x$, and integrate the dynamics on the periodic square \([0,4\pi]\times[0,4\pi]/\{0,4\pi\}\) using a uniform \(80\times 80\) grid. The time step is chosen as \(\Delta t=1.5\times 10^{-3}\), and the evolution is computed over \(140\) steps. The  procedure used to generate the simulations is summarized in Algorithm~\ref{alg:reduced_modified_euler}.

\begin{algorithm}[h!]
\caption{Modified Euler dynamics}
\label{alg:reduced_modified_euler}
\begin{algorithmic}[1]
\State Choose the barrier potential $V(x,y)$ and the initial discrete divergence-free field $X^0$
\For{$n=0,\dots,N-1$}
    \State Update the velocity field by a finite-difference discretization of the modified Euler equation \eqref{mod:eulereq}
    ignoring the divergence-free constraint at the intermediate stage
    \State Denote the resulting intermediate field by $X^\ast$
    \State Compute $\nabla_d\cdot X^\ast$
    \State Solve $\Delta_d \phi = \nabla_d\cdot X^\ast$
    \State Set $X^{n+1} = X^\ast - \nabla_d \phi$
     so that $X^{n+1}$ is approximately discrete divergence-free
\EndFor
\end{algorithmic}
\end{algorithm}

\begin{figure}[h!]
    \centering
    \includegraphics[scale=0.58]{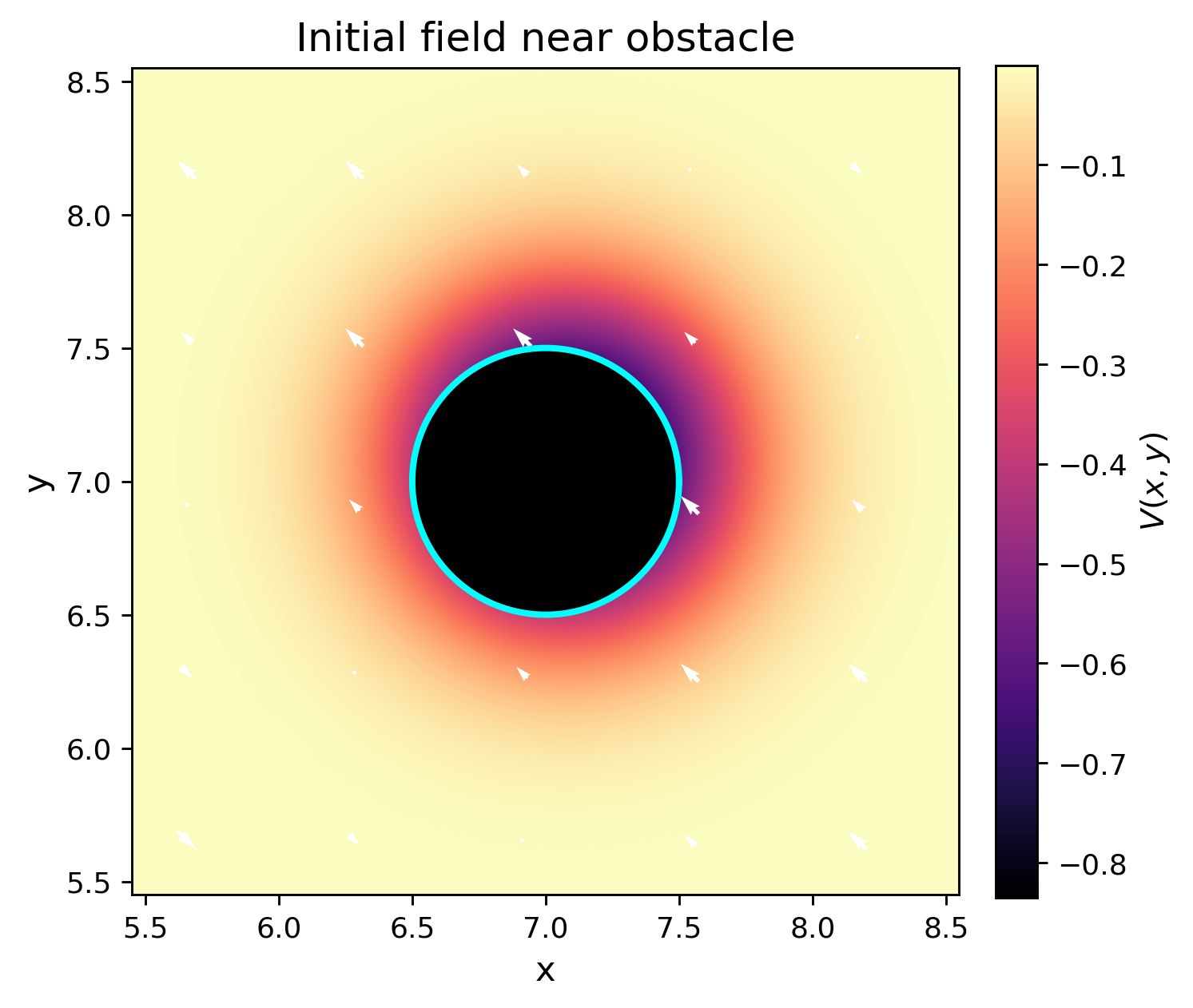}
    \caption{Initial velocity field in the neighborhood of the obstacle. The arrows represent the direction and magnitude of the initial velocity, while the color map shows the potential \(V(x,y)\).}
    \label{fig:initial_flow}
\end{figure}

Figure~\ref{fig:initial_flow} shows the initial velocity field in the neighborhood of the obstacle together with the artificial potential. The arrows represent the local direction and magnitude of the initial velocity, while the color map represents the scalar field \(V(x,y)\). The dark region corresponds to the most negative values of the potential, concentrated near the obstacle, whereas lighter tones correspond to values closer to zero away from it. Since the dynamics depends on \(\nabla V\), this figure should be interpreted as a visualization of the local shaping landscape rather than of the force itself.

\begin{figure}[h!]
    \centering
    \includegraphics[scale=0.58]{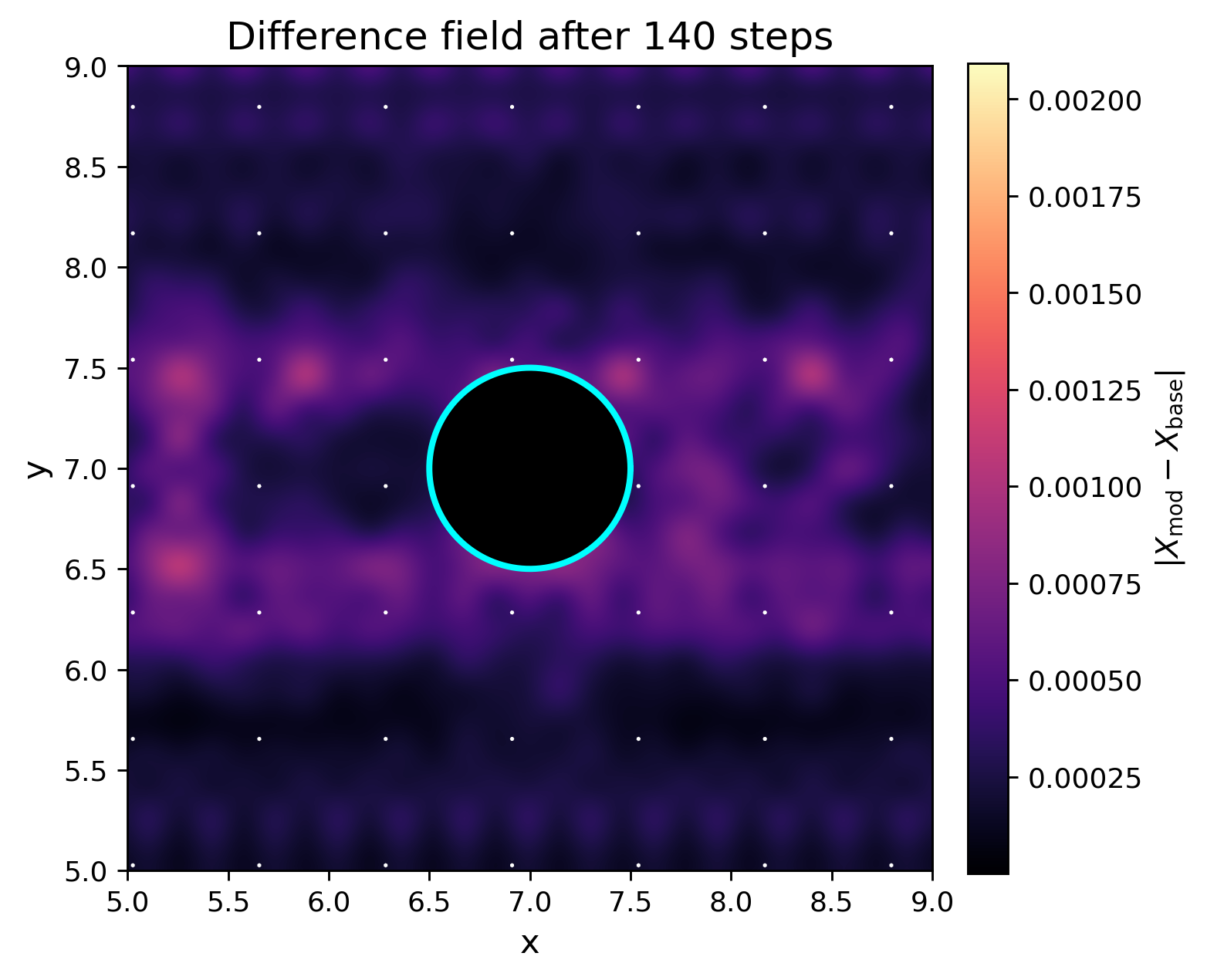}
    \caption{Local difference field between the projected modified and baseline velocity fields after \(140\) time steps. The color map represents the magnitude \(\|X_{\mathrm{mod}}-X_{\mathrm{base}}\|\), showing that the barrier term induces a localized deformation of the flow near the obstacle region.}
    \label{fig:final_flow}
\end{figure}


Figure~\ref{fig:final_flow} displays the local difference field between the projected modified and baseline velocity fields after \(140\) time steps. 
Here, \(X_{\mathrm{base}}\) denotes the projected discrete velocity field obtained from the baseline incompressible Euler evolution, that is, from the numerical scheme without the barrier potential, while \(X_{\mathrm{mod}}\) denotes the projected discrete velocity field obtained from the modified evolution \eqref{mod:eulereq}, including the barrier term. Thus, the figure represents the magnitude
$\|X_{\mathrm{mod}}-X_{\mathrm{base}}\|$ together with the corresponding difference vectors. The color map therefore measures the local sensitivity of the flow to the barrier term: darker regions indicate a smaller deviation from the baseline evolution, whereas lighter regions indicate a stronger local effect of the potential.

The numerical values are consistent with this interpretation. In our computations, the maximum pointwise difference between the two projected fields is of order $\max \|X_{\mathrm{mod}}-X_{\mathrm{base}}\|
\approx
2.09\times 10^{-3}$, with componentwise deviations $\max |\delta u| \approx 1.54\times 10^{-3}$ and $\max |\delta v| \approx 2.09\times 10^{-3}$. Thus, the barrier term does not drastically reorganize the global flow, but it does induce a visible and localized deformation in the neighborhood of the obstacle. This is precisely the effect illustrated in Figure~\ref{fig:final_flow}.

Regarding incompressibility, the discrete Helmholtz projection substantially reduces the divergence error in the baseline evolution, from an initial value of approximately $2.33\times 10^{-2}$ to
$2.74\times 10^{-3}$ at the final time. For the modified evolution, the final discrete divergence remains of size $1.83\times 10^{-2}$, which indicates that the present discretization should be understood as an illustrative numerical approximation rather than as a structure-preserving scheme. Nevertheless, the simulation is sufficient to show that the barrier-type term produces a coherent local deformation of the reduced Eulerian dynamics near the obstacle region.

\section{Conclusions and Future Work}

We have shown that a variational optimal control formulation for incompressible ideal flow with a barrier-type potential leads to modified Euler equations. The potential acts on the Lagrangian configuration as an obstacle-avoidance constraint and appears in the reduced Eulerian description as a shift in the effective pressure.

Our results also suggest that pressure may serve as a natural control mechanism in this setting. In particular, since the barrier term appears in the reduced equations as a pressure shift, one may expect that boundary control of the pressure could be used to induce obstacle-avoidance fluid configurations. A related direction is a more detailed study of the potential $V$ itself, including how its shape, localization, regularity, and strength influence the resulting flow deformation and the structural properties of the reduced dynamics.

Since the present formulation is open-loop, a natural direction for future work is to explore alternatives to barrier-type obstacle avoidance in incompressible flow. Rather than using a configuration-dependent potential, one may study localized steering laws acting directly on the reduced Eulerian dynamics, for example through gyroscopic and dissipative terms supported near the obstacle region. This could lead to a different variational framework for obstacle-aware flow, with a clearer separation between steering and dissipation

\bibliographystyle{IEEEtran}
\bibliography{cluster}

@book{marsden_fluids,
  author = {A. J. Chorin and J. E. Marsden},
  title = {A Mathematical Introduction to Fluid Mechanics},
  year = {1990},
  publisher = {Springer New York, NY},
  series = {Texts in Applied Mathematics},
  edition = {2},
  doi = {10.1007/978-1-4684-0364-0},
  isbn = {978-1-4684-0364-0},
  note = {Originally published in the series: Universitext},
  pages = {IX, 168}
}

@article{PAVLOV,
title = {Structure-preserving discretization of incompressible fluids},
journal = {Physica D: Nonlinear Phenomena},
volume = {240},
number = {6},
pages = {443-458},
year = {2011},
issn = {0167-2789},
author = {D. Pavlov and P. Mullen and Y. Tong and E. Kanso and J.E. Marsden and M. Desbrun}
}

@ARTICLE{HDD,
  author={Bhatia, Harsh and Norgard, Gregory and Pascucci, Valerio and Bremer, Peer-Timo},
  journal={IEEE Transactions on Visualization and Computer Graphics}, 
  title={The Helmholtz-Hodge Decomposition—A Survey}, 
  year={2013},
  volume={19},
  number={8},
  pages={1386-1404},
  doi={10.1109/TVCG.2012.316}}

@ARTICLE{BlCaCoIJC,
  title={Dynamic interpolation for obstacle avoidance on Riemannian manifolds},
  author={A. Bloch and M. Camarinha and L. J. Colombo},
  journal={International Journal of Control},
  volume={94},
  number={3},
  pages={588--600},
  year={2021},
  publisher={Taylor and Francis Ltd}
}

@ARTICLE{sh,
  title={Variational collision and obstacle avoidance of multi-agent systems on Riemannian manifolds},
  author={R.S. Chandrasekaran and L. J. Colombo and M. Camarinha and R. Banavar and A. Bloch},
  journal={Proceedings of the 2020 European Control Conference},
  volume={},
  number={},
  pages={},
  year={2020},
  publisher={}
}

@ARTICLE{mishal,
  title={Variational collision avoidance problems on Riemannian manifolds},
  author={M. Assif and R. Banavar and A. Bloch and M. Camarinha and L. J. Colombo},
  journal={Proceedings of the 2018 IEEE International Conference on Decision and Control},
  volume={},
  number={},
  pages={2791--2796},
  year={2018},
  publisher={IEEE}
}

@ARTICLE{point,
  title={Variational point-obstacle avoidance on Riemannian manifolds},
  author={A. Bloch and M. Camarinha and L. J. Colombo},
  journal={Mathematics of Control, Signals, and Systems},
  volume={33},
  pages={109--121},
  year={2021},
  publisher={Springer London}
}

@ARTICLE{BlCaCoCDC,
  title={Variational obstacle avoidance on Riemannian manifolds},
  author={A. Bloch and M. Camarinha and L. J. Colombo},
  journal={Proceedings of the 2017 IEEE International Conference on Decision and Control},
  volume={},
  number={},
  pages={146--150},
  year={2017},
  publisher={IEEE}
}

@inproceedings{bloch1998discrete,
  title={Discrete rigid body dynamics and optimal control},
  author={Bloch, Anthony M and Crouch, Peter E and Marsden, Jerrold E and Ratiu, Tudor S},
  booktitle={Proceedings of the 37th IEEE Conference on Decision and Control (Cat. No. 98CH36171)},
  volume={2},
  pages={2249--2254},
  year={1998},
  organization={IEEE}
}

@article{bloch1997double,
  title={Double bracket equations and geodesic flows on symmetric spaces},
  author={Bloch, Anthony M and Brockett, Roger W and Crouch, Peter E},
  journal={Communications in mathematical physics},
  volume={187},
  pages={357--373},
  year={1997},
  publisher={Springer}
}

@article{bloch1996optimal,
  title={Optimal control and geodesic flows},
  author={Bloch, Anthony M and Crouch, Peter E},
  journal={Systems \& control letters},
  volume={28},
  number={2},
  pages={65--72},
  year={1996},
  publisher={Elsevier}
}

@article{ebin1969groups,
  title={Groups of diffeomorphisms and the solution of the classical Euler equations for a perfect fluid},
  author={Ebin, David G and Marsden, Jerrold E},
  year={1969}
}

@article{marsden1984semidirect,
  title={Semidirect products and reduction in mechanics},
  author={Marsden, Jerrold E and Ra{\c{t}}iu, Tudor and Weinstein, Alan},
  journal={Transactions of the american mathematical society},
  volume={281},
  number={1},
  pages={147--177},
  year={1984}
}

@article{marsden1983coadjoint,
  title={Coadjoint orbits, vortices, and Clebsch variables for incompressible fluids},
  author={Marsden, Jerrold and Weinstein, Alan},
  journal={Physica D: Nonlinear Phenomena},
  volume={7},
  number={1-3},
  pages={305--323},
  year={1983},
  publisher={Elsevier}
}

@article{gay2013geometric,
  title={Geometric dynamics of optimization},
  author={Gay-Balmaz, Fran{\c{c}}ois and Holm, DD and Ratiu, TS},
  journal={Communications in Mathematical Sciences},
  volume={11},
  number={1},
  pages={163--231},
  year={2013}
}

@book{marsden1993basic,
  title={Basic multivariable calculus},
  author={Marsden, Jerrold E and Tromba, Anthony J and Weinstein, Alan and others},
  year={1993},
  publisher={Springer}
}

@article{holm2009euler,
  title={Euler's fluid equations: Optimal control vs optimization},
  author={Holm, Darryl D},
  journal={Physics Letters A},
  volume={373},
  number={47},
  pages={4354--4359},
  year={2009},
  publisher={Elsevier}
}

@inproceedings{bloch2000optimal,
  title={An optimal control formulation for inviscid incompressible ideal fluid flow},
  author={Bloch, Anthony M and Holm, Darryl D and Crouch, Peter E and Marsden, Jerrold E},
  booktitle={Proceedings of the 39th IEEE Conference on Decision and Control (Cat. No. 00CH37187)},
  volume={2},
  pages={1273--1278},
  year={2000},
  organization={IEEE}
}

@inproceedings{arnold1966geometrie,
  title={Sur la g{\'e}om{\'e}trie diff{\'e}rentielle des groupes de Lie de dimension infinie et ses applications {\`a} l'hydrodynamique des fluides parfaits},
  author={Arnold, Vladimir},
  booktitle={Ann. de l'institut Fourier},
  volume={16},
  pages={319--361},
  year={1966}
}

@article{holm1998euler,
  title={The Euler--Poincar{\'e} equations and semidirect products with applications to continuum theories},
  author={Holm, Darryl D and Marsden, Jerrold E and Ratiu, Tudor S},
  journal={Advances in Mathematics},
  volume={137},
  number={1},
  pages={1--81},
  year={1998},
  publisher={Elsevier}
}

@article{gaybalmaz2011clebsch,
  title={Clebsch optimal control formulation in mechanics},
  author={Gay-Balmaz, Fran{\c{c}}ois and Ratiu, Tudor S},
  journal={J. Geom. Mech},
  volume={3},
  number={1},
  pages={41--79},
  year={2011}
}

@ARTICLE{cluster,
title={Dynamic control of mobile multirobot systems: the cluster space formulation},
author={I.\ \!Mas and C.\ \!Kitts},
journal={IEEE Access},
year={2014},
volume={\!\!2},
number={\!\!2},
pages={\!\!558-570}
}

\section*{Appendix: Auxiliary Lemmas}

\begin{lemma}\label{aux:lemma:1}
    Let $\varphi:\Omega \times \R \to \Omega$ be a diffeomorphism for each fixed $t$, denote
    $v=\frac{\partial \varphi}{\partial t}\circ\varphi^{-1},$
    and let $z=\pi\circ\varphi^{-1}$, where $\pi:\Omega\times\R\to\R^3$. Then $$(T\pi\circ\varphi^{-1})\left(\frac{\partial\varphi^{-1}}{\partial t}\right)=-Tz(v).
    $$
\end{lemma}

\begin{proof}
    First, let us write an alternative expression for $\frac{\partial\varphi^{-1}}{\partial t}$. Indeed, noting that $\varphi\circ\varphi^{-1}=\mathrm{id}_{\Omega}$,
    we deduce that
    $$
    \frac{\partial\varphi}{\partial t}\circ\varphi^{-1}
    +
    (T\varphi\circ\varphi^{-1})\left(\frac{\partial\varphi^{-1}}{\partial t}\right)
    =0.
    $$
    Therefore,
    $(T\varphi\circ\varphi^{-1})\left(\frac{\partial\varphi^{-1}}{\partial t}\right)=-v
    $
    and, hence,
    $$
    \frac{\partial\varphi^{-1}}{\partial t}
    =
    -(T\varphi\circ\varphi^{-1})^{-1}(v)
    =
    -T\varphi^{-1}(v).
    $$
    Finally, applying the map $(T\pi\circ\varphi^{-1})$, we deduce
    \begin{equation*}
        \begin{split}
            (T\pi\circ\varphi^{-1})\left(\frac{\partial\varphi^{-1}}{\partial t}\right)
            &=
            -(T\pi\circ\varphi^{-1})(T\varphi^{-1}(v)) \\
            &=
            -T(\pi\circ\varphi^{-1})(v)
            =
            -Tz(v).
        \end{split}
    \end{equation*}
\end{proof}

\begin{lemma}\label{aux:lemma:2}
    Let $z:\Omega\times\R\to\R^3$ and $v:\Omega\times\R\to\R^3$ be two time-dependent vector fields on $\Omega\subseteq\R^3$. Then
    \begin{enumerate}
        \item $(Tv)^Tz=(z\cdot\grad)v+z\times(\curl v)$.
        \item $Tz(v)=(v\cdot\grad)z$.
        \item $(v\cdot\grad)z+(z\cdot\grad)v+z\times(\curl v)=\grad(v\cdot z)-v\times(\curl z)$.
    \end{enumerate}
\end{lemma}

\begin{proof}
    \begin{enumerate}
        \item Let us compute in coordinates $(z\cdot\grad)v+z\times(\curl v)$ by denoting
        $z=(z_1,z_2,z_3)$, $v=(v_1,v_2,v_3)$,
        and the coordinates in $\R^3$ by $(x_1,x_2,x_3)$. Then
        $$
        (z\cdot\grad)v
        =
        \left(
        z_i\frac{\partial v_1}{\partial x_i},
        z_i\frac{\partial v_2}{\partial x_i},
        z_i\frac{\partial v_3}{\partial x_i}
        \right)
        $$
        and
        \begin{equation*}
            \begin{split}
                z\times(\curl v)
                =
                &\left(
                z_2\left(\frac{\partial v_2}{\partial x_1}-\frac{\partial v_1}{\partial x_2}\right)
                -
                z_3\left(\frac{\partial v_1}{\partial x_3}-\frac{\partial v_3}{\partial x_1}\right),
                \right.\\
                &
                z_3\left(\frac{\partial v_3}{\partial x_2}-\frac{\partial v_2}{\partial x_3}\right)
                -
                z_1\left(\frac{\partial v_2}{\partial x_1}-\frac{\partial v_1}{\partial x_2}\right),\\
                &
                \left.
                z_1\left(\frac{\partial v_1}{\partial x_3}-\frac{\partial v_3}{\partial x_1}\right)
                -
                z_2\left(\frac{\partial v_3}{\partial x_2}-\frac{\partial v_2}{\partial x_3}\right)
                \right).
            \end{split}
        \end{equation*}
        Once we add both vectors, most of the terms cancel out and we are left with
        \begin{equation*}
            \begin{split}
                (z\cdot\grad)v+z\times(\curl v)
                &=
                \left(
                z_i\frac{\partial v_i}{\partial x_1},
                z_i\frac{\partial v_i}{\partial x_2},
                z_i\frac{\partial v_i}{\partial x_3}
                \right) \\
                &=
                (Tv)^Tz.
            \end{split}
        \end{equation*}

        \item It is straightforward, just notice that
        \begin{align*}
            Tz(v)
            =
            &
            \left(
            v_i\frac{\partial z_1}{\partial x_i},
            v_i\frac{\partial z_2}{\partial x_i},
            v_i\frac{\partial z_3}{\partial x_i}
            \right)=
            (v\cdot\grad)z.
        \end{align*}

        \item This item can be found in the literature (see \cite{marsden1993basic}, for instance).
    \end{enumerate}
\end{proof}

\end{document}